\documentclass[letterpaper,11pt,prl,reprint,aps,showpacs,nofootinbib]{revtex4-2}

\usepackage{amsmath}
\usepackage{amssymb}
\usepackage{graphicx}
\usepackage[colorlinks=true, allcolors=blue]{hyperref}
\usepackage{amsthm}
\usepackage{mathtools}
\usepackage{appendix}
\usepackage{tikz-cd}
\usepackage[all,cmtip]{xy}

\newtheorem{definition}{Definition}
\newtheorem{theorem}{Theorem}
\newtheorem{lemma}{Lemma}
\newtheorem{proposition}{Proposition}
\newtheorem{corollary}{Corollary}

\definecolor{FGreen}{RGB}{1,100,10}

\newcommand{\drawgenerator}[8]{%
\xymatrix@!0{%
& #8 \ar@{-}[ld]\ar@{.}[dd] \ar@{-}[rr] & & #7 \ar@{-}[ld]  \\%
#1 \ar@{-}[rr] \ar@{-}[dd] &  & #2 \ar@{-}[dd] &            \\%
& #6 \ar@{.}[ld] &  & #5 \ar@{-}[uu] \ar@{.}[ll]       \\%
#3 \ar@{-}[rr] &  & #4 \ar@{-}[ru]                       %
}}

\newcommand{\plaquette}[4]{
\xymatrix@!0{%
#1 \ar@{-}[r] \ar@{-}[d]  & #2 \ar@{-}[d] 
\\
#3 \ar@{-}[r]  & #4
}}

\begin{document}

\title{
A Classification of Translation-Invariant Quantum Codes in Any Dimension
}
\author{Andrew Li}
\author{Dominic J. Williamson}
\affiliation{School of Physics, University of Sydney, Sydney, NSW 2006, Australia}
\date{Aug 2026}

\begin{abstract}
\noindent
Quantum error-correcting codes with two-dimensional translation invariance are known to be equivalent to copies of the two-dimensional toric code. Such a simple classification is not possible for quantum codes with higher dimensional translation invariance due to the existence of multiple types of toric codes and infinite families of fracton codes. Here, we focus on $D$-dimensional translation-invariant quantum codes based on length-$D$ chain complexes. This includes multivariate multicycle codes where the number of variables equals the number of cycles. We show that such codes are equivalent to copies of $D$-dimensional toric codes. This directly generalizes the classification result for two-dimensional translation-invariant codes. 
\end{abstract}

\maketitle

\vspace{1cm}

%%%%%%%%%%%%%%%%%%%%%%%%%%%%%%%%%%%%%%%%%%%%%%%%%%%%%%%%%%%%%%%%%%%%
Quantum computers rely on quantum error-correcting codes to perform useful computations fault tolerantly in the presence of realistic noise. 
The surface code has underpinned standard fault-tolerant quantum architectures for decades~\cite{Kitaev_2003,bravyi1998quantumcodeslatticeboundary,Dennis_2002}. 
Research on the surface code, and other topological codes, has generated a productive dialogue between the quantum error correction and quantum phases of matter communities~\cite{Kitaev_2006,Nayak_2008,Levin_2005,Koenig_2010}. 
Recently, quantum low-density parity-check (qLDPC) codes that generalize the surface code by allowing long range connectivity have become the focus of intense activity due to their promising performance~\cite{Breuckmann_2021}. 

Translation-invariant (TI) qLDPC codes, including bivariate bicycle codes, possess a symmetry structure that simplifies their implementation and fault-tolerant logic~\cite{Bravyi2024,PhysRevA.88.012311,Panteleev_2021,wang2023abeliannonabelianquantumtwoblock,Haah_2017,KalPan20}. 
All two-dimensional TI codes with fixed check structure and a growing distance are equivalent to copies of the toric code up to finite coarse graining, the additional of ancillary qubits in a product state, and the application of a local unitary circuit~\cite{Haah_2013}. 
This structure has been applied to characterize two-dimensional TI qLDPC code properties in terms of symmetry-enriched toric code anyons~\cite{Dua_2019,Liang_2024,Chen_2025,Liang_2025,hopkin2026translationinvariantquantumlowdensityparitycheck,wang2026decoupling2dtranslationinvarianttopological}. 
In higher dimensions, such a simple classification result is not possible due to the existence of multiple types of toric codes and infinite classes of inequivalent fracton codes~\cite{Nandkishore_2019,Pretko_2020,Dua_2019sorting}. 

In this work we consider $D$-dimensional translation-invariant (DDTI) quantum codes that are derived from length-$D$ chain complexes. This is a natural condition guaranteeing the dimension of a code's chain complex matches the dimension of its translation group. 
This class of codes includes all $D$-variate $D$-cycle codes and does not include fracton codes which have a smaller chain complex length than translation group dimension. 
We establish a structure theorem for such DDTI codes, by showing that they are equivalent to copies of $D$-dimensional toric codes up to finite coarse graining, the additional of ancillary qubits in a product state, and the application of a local unitary circuit. 
This establishes a full classification of such DDTI codes in terms of copies of $D$-dimensional toric codes. 
We note that in dimension 4 and above, there are multiple inequivalent types of toric code which may appear in the classification. 

\begin{figure}
    \centering
    \includegraphics[width=1\linewidth]{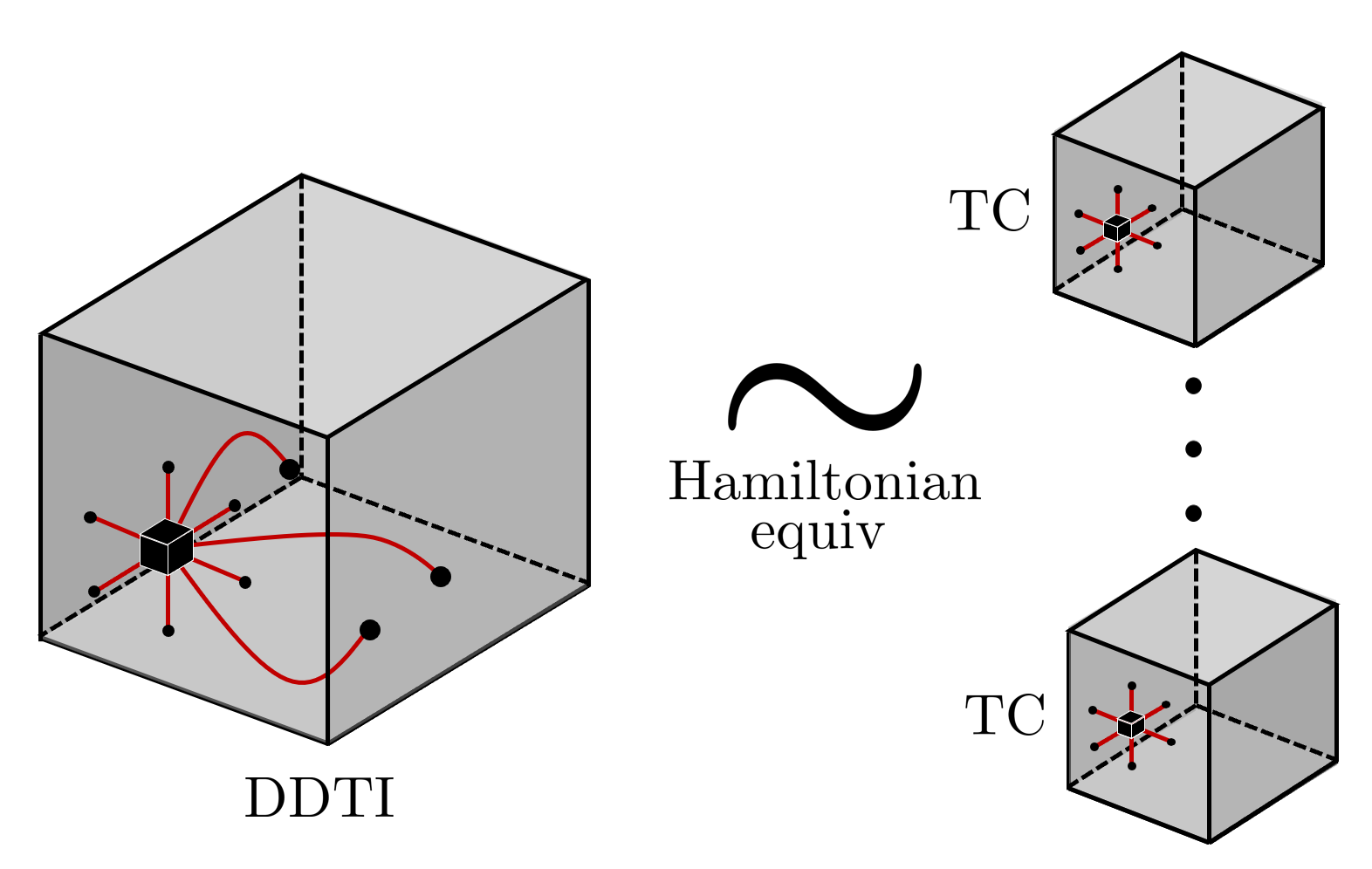}
    \caption{A depiction of our main result.
    Any $D$-dimensional translation-invariant (DDTI) qLDPC code based on a length-$D$ chain complex, satisfying the conditions in Theorem~\ref{ham_equivalence}, is equivalent to copies of $D$-dimensional toric codes (TC). 
    This includes all multivariate multicycle codes with an equal number of variables and cycles.}
    \label{fig:decoupling}
\end{figure}

%%%%%%%%%%%%%%%%%%%%%%%%%%%%%%%%%%%%%%%%%%%%%%%%%%%%%%%%%%%%%%%%%%%%

\noindent\textbf{CSS codes as chain complexes} ---
A chain complex $C_\bullet$ is a sequence of vector spaces (or modules) $C_n$ indexed by the integers $\mathbb{Z}$ equipped with differentials ${d_n : C_n \to C_{n-1}}$ such that $d_{n-1} d_{n} = 0$. A chain complex is bounded if only finitely many $C_n$ are nonzero, and its length is the difference between the minimum and maximum degrees $n$ at which $C_n$ is nonzero. One can describe a CSS code \cite{PhysRevA.54.1098, PhysRevA.54.4741} as a length-2 chain complex of $\mathbb{F}_2$ vector spaces:
\begin{align}
    C_2 \xrightarrow[]{d_2} C_1 \xrightarrow[]{d_1} C_0
\end{align}
where $C_2$ admits a basis labelled by $Z$ stabiliser checks, $C_1$ labelled by the qubits of the code, and $C_0$ labelled by the $X$ stabiliser checks. The differentials $d_2$ and $d_1$ are linear maps that send $Z$ stabilisers to their qubit support and the qubit support of $Z$ operators to their $X$ stabiliser syndromes, respectively. The $Z$ logical operators are given by the 1-cycles $Z_1 := \ker d_1$ and the $Z$ stabilisers are given by the 1-boundaries $B_1 := \mathrm{im}\, d_2$. The $Z$ logical classes are thus described by the homology group $H_1 := Z_1 / B_1$.
The cochain complex $C^\bullet$ is:
\begin{align}
    C^2 \xleftarrow[]{d_2^*} C^1 \xleftarrow[]{d_1^*} C^0    
\end{align}
where for finitely many qubits, $C^i \cong C_i$; $d_1^*$ takes X stabilisers to their support and $d_2^*$ takes an X operator to its Z stabiliser syndrome. $Z^1 := \ker d_2^*$, $B^1 := \mathrm{im} \,d_1^*$, and $H^1 := Z^1/B^1$ respectively describe the X logicals, X stabilisers, and X logical classes.

\noindent\textbf{Translation-invariant codes} ---
We consider the $D$-dimensional lattice $\mathbb{Z}^D$, with $q$ qubits at each site. We then enforce stabilisers on this lattice of qubits; the resulting code is translation-invariant (TI) if all translations of a stabiliser are themselves stabilisers. If the code is CSS, its CSS chain complex is a sequence of finitely-generated free $R$-modules over $R=\mathbb{F}_2[x_1^{\pm 1}, \dots, x_D^{\pm 1}]$:
\begin{align}
    F_2 \xrightarrow[]{d_2} F_1 := R^q \xrightarrow[]{d_1} F_0
\end{align}
The translational invariance of the stabilisers is enforced by the $R$-linearity of $d_2$ and $d_1$. 
Here, the monomial $x_1^{i_1}\cdots x_D^{i_D}$ indicates the site at coordinate $(i_1,\dots,i_D)$.
We also consider the complex $F^\bullet$:
\begin{align}
    F^2 \xleftarrow[]{d_2^\dagger} F^1 \xleftarrow[]{d_1^\dagger} F^0    
\end{align}
with $F^0 = F_0$, etc.
We call this the ZX-dual complex; although we overload notation, note that it is not the same as the cochain complex (considering $F_i$ as vector spaces).
It is typical to demand that both complexes are exact at $F_1$ and $F^1$ respectively; this ensures that there are no finite-weight logicals and thus when we compactify to a finite lattice the distance grows with the lattice size. 
Any exact CSS complex can be embedded into a bounded free resolution of some $R$-module $M$; by 
this we mean a bounded complex of free modules which when extended by adding a differential to $M$ becomes exact. 

If we can extend the complex to the right, this gives rise to local $X$ metachecks, and if we can extend it to the left, we find local $Z$ metachecks. Metachecks refer to nontrivial collections of checks that multiply to the identity. 
In particular, one can obtain the infinite-lattice counterpart of multivariate multicycle (MVMC) codes by extending the presentation 
$$R^k \xrightarrow[]{[f_1,\dots,f_k]} R \to R/(f_1,\dots f_k),$$ 
where $(f_1,\dots f_k)$ is a regular sequence (this formalizes the notion that each $f_i$ should be independent), into a resolution and then taking a length-2 segment of the resolution \cite{mian2026multivariatemulticyclecodescomplete}. Such a resolution is called a Koszul complex.

\begin{definition}
    The CSS complex of an MVMC code on an infinite lattice is given by a length-2 segment, at degrees $n+1$ to $n-1$ (where $1 \leq n < k$), of the Koszul complex derived from a regular sequence $(f_1,\dots,f_k)$. If $k=D$, we call it a $D$-variate $D$-cycle code.
\end{definition}

Note that we cannot have more cycles than variables. One such family of codes obtained from this construction are the $D$-dimensional toric codes \cite{Dennis_2002}:

\begin{definition}
    For $1 \leq n < D$, the infinite-lattice $(n,D-n)$-toric code is obtained from the exact CSS complex $F_{n+1} \to F_n \to F_{n-1}$ where $F_\bullet$ is the Koszul complex resolving $R/(x_1 - 1, \dots, x_D - 1)$. Compactifying to a finite lattice yields an $(n,D-n)$-toric code.
\end{definition}

In particular, the checks of the 2D toric code are:
\begin{align}
    d_1^\dagger=
    \begin{pmatrix}
        1+\bar{x} \\
        1+\bar{y}
    \end{pmatrix}
\, ,
\quad d_2=
    \begin{pmatrix}
        1 + y \\
        1 + x
    \end{pmatrix}
    \, .
\end{align}
Explicitly, the translation variables and checks appear as:
\begin{align}
\begin{array}{c}
\plaquette{y}{xy}{1}{x}
\qquad
\plaquette{XI}{XX}{II}{IX}
\qquad\plaquette{ZI}{II}{ZZ}{IZ}
\end{array}
\, .
\label{unit_cell2D}
\end{align}
Here, the left-most cell labels the $\mathbb{Z}^2$-lattice coordinates for reference. As the Koszul complex has length 2, there is only one way to fit the CSS complex into the resolution, and thus only one choice of 2D toric code.

The $X$ checks, $Z$ checks, and $Z$ metachecks of the 3D toric code are respectively given by:
\begin{align}
    d_1^\dagger = 
    \begin{pmatrix}
        1+\bar{x} \\
        1 + \bar{y} \\
        1+ \bar{z}
    \end{pmatrix}
    d_2 =
    \begin{pmatrix}
        0 & 1+z & 1+y  \\
        1+z & 0 & 1+x \\
        1+y & 1+x & 0
    \end{pmatrix}
    d_3=
    \begin{pmatrix}
        1 + x \\
        1 + y \\
        1+ z
    \end{pmatrix}
\end{align}
The $\mathbb{Z}^3$-lattice coordinates and metachecks are depicted:
\begin{align}
\begin{array}{c}
\drawgenerator{xz}{xyz}{x}{xy}{y}{1}{yz}{z}
\end{array}
\qquad 
\begin{array}{c}
\drawgenerator{}{}{s^Z_{\hat{x}}}{}{s^Z_{\hat{y}}}{s^Z_{\hat{x}}s^Z_{\hat{y}}s^Z_{\hat{z}}}{}{s^Z_{\hat{z}}}
\end{array}
\, ,
\label{unit_cell3D}
\end{align}
and similarly for the checks $s^X,s_{\hat{x}}^Z,s_{\hat{y}}^Z,s_{\hat{z}}^Z$ in a unit cell :
\begin{align}
\begin{array}{c}
\drawgenerator{IXI}{XXX}{}{IIX}{}{}{XII}{}
\quad
\drawgenerator{}{}{}{}{IIZ}{IZZ}{}{IZI}
\quad
\\
\drawgenerator{}{}{IIZ}{}{}{ZIZ}{}{ZII}
\quad
\drawgenerator{}{}{IZI}{}{ZII}{ZZI}{}{}
\end{array}
\, .
\label{3DTCstabilizers}
\end{align}	 
In particular, this is the $(1,2)$-toric code, occupying the right end of the length-3 resolution, thus the presence of local Z but not X metachecks. The $(2,1)$-toric code which occupies the left end of the resolution is the ZX-dual of the $(1,2)$-toric code and thus we typically speak of `the 3D toric code' as unique up to ZX-duality. For higher dimensions, we get meaningfully distinct toric codes; for instance the $(1,3)$ and $(2,2)$-toric codes. The latter, occupying the centre of the chain complex, has both local X and Z metachecks whereas the former occupies the right and has only local Z metachecks.

\noindent\textbf{Hamiltonian equivalence} --- 
We now review Haah's stabiliser complex and formulation of Hamiltonian equivalence \cite{Haah_2013}.
The stabiliser complex is obtained by taking the direct sum of a CSS complex $F_\bullet$ of a TI code on an infinite lattice with its ZX-dual; that is $F_\bullet \oplus F^\bullet$. The cycles $Z(F_1 \oplus F^1)$ form the stabiliser module. Given two infinite-lattice TI codes, we say that their Hamiltonians are equivalent and lie in the same topological phase if their stabiliser modules become identical after some combination of coarse-graining, tensoring ancillas, and symplectic transformations. We briefly describe each of these below.

First we describe coarse-graining. We can select a sublattice of $\Lambda \subseteq \mathbb{Z}^D$ and then consider the group algebra $R' := \mathbb{F}_2[\Lambda]$, which is a subring of $R$. We then coarse-grain a stabiliser complex by taking the restriction of scalars to $R'$; that is we only allow $R'$-multiplication but keep the underlying abelian group that describes the $\mathbb{F}_2$-vector space. 
Consider the chain complex 
$$R \xrightarrow[]{\left[\begin{smallmatrix} 1 \\ 0 \end{smallmatrix}\right]} R^2 \xrightarrow[]{\left[\begin{smallmatrix} 0 & 1 \end{smallmatrix}\right]} R \, ,$$ or its ZX-dual. These correspond to single layers of X or Z product states on the lattice. Taking the direct sum of the stabiliser complex with such a complex is called tensoring an ancilla. A symplectic transformation $T$ is an automorphism of $F_1 \oplus F^1 = R^{2q}$ that preserves the symplectic form, that is $T^\dagger \lambda T = \lambda$ where $\lambda = \left[\begin{smallmatrix} 0 & 1 \\ 1 & 0 \end{smallmatrix}\right]$ in block form.

A well-known result due to Bombín~\cite{Bombin_2012} states that any 2D topological code is equivalent to finitely many copies of the toric code. In Haah's formulation, this is restated as Hamiltonian equivalence between 2D exact stabiliser complexes and finitely many copies of the 2D toric code on the infinite plane~\cite{Haah_2013,Haah_2017}. This implies the existence of a growing family of periodic lattices on which the code, when compactified, is equivalent to a stack of 2D toric codes up to local unitaries and coarse-graining. Below, we describe a generalisation of this result to a class of codes in higher dimensions, including $D$-variate $D$-cycle codes.

%%%%%%%%%%%%%%%%%%%%%%%%%%%%%%%%%%%%%%%%%%%%%%%%%%%%%%%%%%%%%%%%%%%%
\noindent\textbf{Classification of $D$-dimensional Translation-Invariant Quantum Codes} --- We let $R := \mathbb{F}_2[x_1^{\pm 1}, \dots, x_D^{\pm 1}]$ denote positions on the infinite $D$-dimensional hypercubic lattice. Let $F_\bullet$ be a length-$D$ complex of free modules such that $F_\bullet$ and $F^\bullet$ are exact except at their right and left ends respectively. This ensures that if we place a CSS chain complex along any length-2 segment of nonzero modules of $F_\bullet$, it will have growing distance after compactifying to finite lattices. We can frame $F_\bullet$ and $F^\bullet$ as free resolutions of modules, respectively $F_0/ \mathrm{im}\,\partial_1$ and $F^D/\mathrm{im}\,\partial_D^\dagger$. Let us call the former module $M$, so that $F_\bullet$ resolves $M$. Under our conditions, $M$ is a zero-dimensional module, that is its annihilator $\mathrm{ann}_R (M)$ is a zero-dimensional ideal.

\begin{proposition}\label{d-dim-complex}
    Let $F_\bullet$ be a length-$D$ complex of finite-rank free $R$-modules resolving $M$. If $F^\bullet$ is exact at degrees $<D$, then $M$ is a zero-dimensional module.
\end{proposition}

\begin{proof}
    See End Matter.
\end{proof}

That $M$ is zero-dimensional tells us that the topological charges (or metacharges) it describes are freely mobile and fall into finitely many superselection sectors.
We recall a result that demonstrates the existence of a constant scale $L$ over which any charge can be moved: 
\begin{proposition}[{\cite[Lemma 7.3]{Haah_2013}}]\label{haah}
Let $M$ be a nontrivial zero-dimensional module. Then there exists $L \geq 1$ such that:
\begin{align}
    \mathrm{ann}_{R'} (M) = (x_1^L - 1, \dots x_D^L-1)
\end{align}
for $R' = \mathbb{F}_2[x_1^{\pm L}, \dots, x_D^{\pm L}] \subseteq R$
\end{proposition}
In other words, the above proposition guarantees that after coarse graining by a constant factor $L$ any translate of a charge on the coarse grained lattice is equivalent. 

Next, we state a lemma characterizing the charges on the coarse grained lattice. 
\begin{lemma}\label{coarse_grain}
    Let $N$ be the restriction of scalars of $M$ to $R'$, and $F'_\bullet$ the complex attained by taking the restriction of scalars of each $F_i$. Then $F'_\bullet \to N$ is a free resolution of $N$. Moreover, $N \cong (R'/\mathfrak{m})^{\oplus k}$ where ${\mathfrak{m} = (x_1^L - 1, \dots x_D^L-1)}$ and $k = \mathrm{dim}_{\mathbb{F}_2}(M)$.
\end{lemma}
This lemma implies that the charge superselection sectors on the coarse-grained lattice can be generated by $k$ independent generators which are each equivalent to their translates on the coarse-grained lattice. 

\begin{proof}
    As $R'$-modules, $R \cong R'^b$ where $b = L^D$, so it is readily seen that each $F_i'$ is free. As exactness is unchanged by restriction, it follows that $F_\bullet' \to N$ is a free resolution. Now, note that $R'/\mathfrak{m} \cong \mathbb{F}_2$. 
    As $\mathrm{ann}_{R'} (N) = \mathfrak{m}$ by Proposition.~\ref{haah}, we have that the $R'$-linear maps and $R'/\mathfrak{m}$ or $\mathbb{F}_2$-linear maps on $N$ coincide. Combining that $N$ is a finitely-generated $R'$-module and $\mathrm{ann}_{R'}(N) = \mathfrak{m}$, we see $N$ is a finite-dimensional $\mathbb{F}_2$-vector space (A generating set of $N$ also serves as a $\mathbb{F}_2$-spanning set) and thus letting its dimension be $k$, the invertible $\mathbb{F}_2$-linear map between $N$ and $\mathbb{F}_2^k$ also serves as an isomorphism between $N$ and $(R'/\mathfrak{m})^{\oplus k}$.
\end{proof}

We next state a lemma that establishes equivalence between different finite chain complexes that provide free resolutions of the same charge module. 
\begin{lemma}\label{tensor_ancilla_symplectic}
    Given two bounded free resolutions $F_\bullet \to M$ and $G_\bullet \to M$ of a finitely-generated $R$-module $M$ by finite-rank free modules, there exist complexes $C_\bullet$ and $C_\bullet'$ such that $F_\bullet + C_\bullet \cong G_\bullet + C_\bullet'$ and each can be written as a direct sum of complexes in varying degrees of the form $0 \to R^n \xrightarrow[]{\mathrm{id}} R^n \to 0$.
\end{lemma}
\begin{proof}
    See End Matter.
\end{proof}
The equivalence between $F_\bullet$ and $G_\bullet$ in this lemma implies that the codes derived from length-2 segments, $F_{n+1}\rightarrow F_n \rightarrow F_{n-1}$ and $G_{n+1}\rightarrow G_n \rightarrow G_{n-1}$, are equivalent up to the addition of ancillary qubits in a product state and the application of a local unitary circuit. 

We can now state our main result, which follows directly from the preceding discussion.

\begin{theorem}\label{ham_equivalence}
    Consider the CSS complex ${F_{n+1} \to F_{n} \to F_{n-1}}$ obtained from a length-$D$ chain complex satisfying the conditions of Proposition.~\ref{d-dim-complex}, with $1 \leq n < D$. Its corresponding code Hamiltonian is equivalent to $k$ copies of the $(n,D-n)$-toric code Hamiltonian on an infinite lattice, where $k = \mathrm{dim}_{\mathbb{F}_2}(M)$.
\end{theorem}
\begin{proof}
    We want to show that the stabiliser modules become equivalent after some sequence of coarse graining, symplectic transformations, and tensoring ancillas. Let $G_\bullet$ be the direct sum of $k$ Koszul resolutions of $R'/\mathfrak{m}$. It is a free resolution of $(R'/\mathfrak{m})^k$, and any three consecutive terms describe the CSS complex of a $D$-dimensional toric code. After coarse-graining, $F_\bullet + C_\bullet =: F_\bullet' \cong G_\bullet ':= G_\bullet + C_\bullet'$ by Lemma.~\ref{coarse_grain} and Lemma.~\ref{tensor_ancilla_symplectic} and where each $C_\bullet$ is a finite direct sum of $0 \to R \xrightarrow[]{\mathrm{id}} R \to 0$ at different locations. Let the isomorphism be $f: F_\bullet' \to G_\bullet'$. Analogously, $F'^\bullet  \cong G'^\bullet$ by the isomorphism $f^\dagger$. Consider the following chain isomorphism:
    \begin{widetext}
    \begin{equation}
    \begin{tikzcd}[column sep=huge, row sep=large]
    F'_{n+1} \oplus F'^{n-1} 
    \arrow[r, "d_{n+1}\oplus d_n^\dagger"] \arrow[d, "f_{n+1} \oplus (f_{n-1}^{-1})^\dagger"']
    &
    F_n' \oplus F'^n 
    \arrow[r, "d_n \oplus d_{n+1}^\dagger"] \arrow[d, "f_n \oplus (f_n^{-1})^\dagger"'] 
    &
    F'_{n-1}\oplus F'^{n+1} 
    \arrow[d, "f_{n-1} \oplus (f_{n+1}^{-1})^\dagger"] \\
    G'_{n+1} \oplus G'^{n-1} 
    \arrow[r, "d_{n+1}\oplus d_n^\dagger"'] 
    &
    G'_n \oplus G'^n 
    \arrow[r, "d_n\oplus d_{n+1}^\dagger"'] 
    &
    G'_{n-1}\oplus G'^{n+1}
    \end{tikzcd}
    \end{equation}
    \end{widetext}

    At the top and bottom are the stabiliser complexes formed from the CSS complexes described by $F'^\bullet$ and $G'^\bullet$ at degrees $n-1$ to $n+1$. It can be checked that $f_n \oplus (f_n^{-1})^\dagger$ is a symplectic transformation that induces an isomorphism between the stabiliser modules of the two complexes. Thus, the two complexes describe equivalent Hamiltonians. Now consider the summand 
    $C_{n+1} \oplus C^{n-1} \to C_n \oplus C^n \to C_{n-1} \oplus C^{n+1}$ 
    of the top complex. It can be further decomposed into summands $$R \xrightarrow[]{\left[\begin{smallmatrix} 1 \\ 0 \end{smallmatrix}\right]} R^2 \xrightarrow[]{\left[\begin{smallmatrix} 0 & 1 \end{smallmatrix}\right]} R, \quad R \xrightarrow[]{\left[\begin{smallmatrix} 0 \\ 1 \end{smallmatrix}\right]} R^2 \xrightarrow[]{\left[\begin{smallmatrix} 1 & 0 \end{smallmatrix}\right]} R,$$
    $R \to 0 \to 0$, and $0 \to 0 \to R.$ The first two types correspond to tensoring ancillas and removing them preserves Hamiltonian equivalence.
    The last two do not affect the stabiliser module and removing them similarly preserves equivalence. We can also remove similar summands from the bottom complex. We thus get that the stabiliser complexes formed from the CSS complexes described by $F_\bullet$ and $G_\bullet$ at degrees $n-1$ to $n+1$ describe equivalent Hamiltonians.
\end{proof}

\begin{corollary}
    Any $D$-variate $D$-cycle code is Hamiltonian-equivalent to $k$ copies of some $D$-dimensional toric code on an infinite lattice.
\end{corollary}

Although MVMC codes of this form are the most obvious class for which our results apply, the free resolutions from which our CSS complexes are taken need not be Koszul in general; indeed, our result admits any length-$D$ finite-rank free resolution whose ZX-dual is also a resolution.

%%%%%%%%%%%%%%%%%%%%%%%%%%%%%%%%%%%%%%%%%%%%%%%%%%%%%%%%%%%%%%%%%%%%
\noindent\textbf{Discussion} --- In this work, we have shown that any DDTI CSS code derived from a length-2 segment of length-$D$ resolution with a ZX-dual resolution, is equivalent to copies of a $D$-dimensional toric code.
In particular, this applies to all multivariate multicycle codes 
where the number of variables equals the number of cycles. This does not however apply to families with more variables than cycles; topological charges are typically not mobile in this case and thus Lemma.~\ref{coarse_grain} does not hold. We remark that it is not possible to meaningfully construct MVMC codes with more cycles than variables.

This work raises a number of questions we leave for future work. 
Our results show that all DDTI CSS codes derived from length-$D$ resolutions with ZX-dual resolutions fall into gapped quantum liquid phases of matter~\cite{Yoshida_2011,Zeng_2015}. 
Does our result provide a full classification of all TI gapped quantum liquid codes?
Can we loosen the conditions required of the codes we consider? In particular, can our results be generalized beyond translation-invariant codes to show that any finite-range qLDPC code in $D$-dimensions described by an irreducible length-$D$ chain complex is equivalent to copies of a $D$-dimensional toric code. 
Similarly, is any length-$K$ chain complex with $K>D$ that is locally embedded into $D$-dimensional space equivalent to a length-$D$ chain complex? 
We note that this holds for any DDTI free resolution of a finitely generated module.
Can the approach used here be extended to the classification of fracton codes? 
Finally, can the structure theorem we have derived be applied to find decoders for higher dimensional translation-invariant codes, following Refs.~\cite{sahay2026matchingdecoderbivariatebicycle,tan2026generalizedmatchingdecoders2d}? 

\section*{Data Availability}
No data was generated during this work. 

\section*{Acknowledgements} 
While this work was in progress, an independent work was posted to the arxiv with overlapping results~\cite{song2026koszulcomplexstabilizermodelssuperselection}. 
This inspired the addition of Proposition~\ref{d-dim-complex} to our work.
DJW is supported by the Australian Research Council Discovery Early Career Research Award (DE220100625). 

\section*{Author contributions}
AL and DJW both contributed extensively to this paper.

\bibliography{main.bib}

\vspace{2cm}

\section*{End Matter}
Here, we prove several basic results that are used in the main text. 

\noindent\textbf{Proof of Lemma.~\ref{tensor_ancilla_symplectic}} ---
Let $R := \mathbb{F}_2[x_1^{\pm 1}, \dots, x_D^{\pm 1}]$.

\begin{proposition}[Comparison theorem, {\cite[Theorem 2.2.6]{Weibel_1994}}]
\label{comparison_theorem}
    Given a projective resolution $P_\bullet \to_\epsilon M$ of an $R$-module $M$ and a resolution $Q_\bullet \to_\eta N$ of a $R$-module $N$, then for a map $f': M \to N$ there exists a chain map $f: P_\bullet \to Q_\bullet$, unique up to chain-homotopy, such that $\eta f_0 = f' \epsilon$.
\end{proposition}

\begin{corollary}
    \label{homotopy_equiv}
    Two projective resolutions $P_\bullet \to_\epsilon M$ and $Q_\bullet \to_\eta M$ of an $R$-module $M$ are chain-homotopy equivalent.
\end{corollary}
\begin{proof}
    From Proposition.~\ref{comparison_theorem} setting $f' = \mathrm{id}_M$ we have chain maps $f: P_\bullet \to Q_\bullet$ and $g: Q_\bullet \to P_\bullet$ such that $\eta f_0 = \epsilon$ and $\epsilon g_0 = \eta$. Composing them, we have $gf: P_\bullet \to P_\bullet$ is a chain map with $\epsilon (g_0 f_0) = \epsilon = \mathrm{id}_M \epsilon$. $id_{P_\bullet}$ is another chain map satisfying the same conditions. Applying Proposition.~\ref{comparison_theorem}, we get that $gf \simeq \mathrm{id}_{P_\bullet}$. Analogously, it can be shown that $fg \simeq \mathrm{id}_{Q_\bullet}$.
\end{proof}

\begin{proposition}[{\cite[Example 2.2.6]{KLEIN2016497}}]
    \label{property_complex}
    Let $\mathcal{A}$ be an exact category. Consider $\mathbf{C}^\mathrm{b}(\mathcal{A})$, the exact category of bounded complexes in $\mathcal{A}$ where the conflations are defined to be the short exact sequences which split degree-wise. $\mathbf{C}^\mathrm{b}(\mathcal{A})$ is a Frobenius category whose stable category $\underline{\mathbf{C}^\mathrm{b}(\mathcal{A})}$ is the bounded homotopy category of $\mathcal{A}$ and whose projective-injective objects are the bounded contractible complexes.
\end{proposition}

\begin{proposition}[{\cite[Lemma 2.1.27]{Krause_2021}}]
    \label{homotopy_to_iso}
    Let $\mathcal{A}$ be a Frobenius category and $\underline{\mathcal{A}}$ be its stable category. For objects $X$ and $Y$, $X \cong Y$ in $\underline{\mathcal{A}}$ if and only if $X \oplus I \cong Y \oplus J$ for injective objects $I, J \in \mathcal{A}$.
\end{proposition}

\begin{lemma}\label{almost_result}
    Given two bounded resolutions $P_\bullet \to M$ and $Q_\bullet \to M$ of a finitely-generated $R$-module $M$ by finitely-generated projective modules, there exist bounded contractible complexes of projective modules $C_\bullet$ and $C_\bullet'$ such that $P_\bullet + C_\bullet \cong Q_\bullet + C_\bullet'$.
\end{lemma}
\begin{proof}
    Consider the Frobenius category of bounded complexes of finitely generated projective $R$-modules $\mathbf{C}^\mathrm{b}(\mathrm{proj}\, R)$, per Proposition.~\ref{property_complex}. Thus its stable category coincides with its bounded homotopy category. As $P_\bullet \simeq Q_\bullet$ from Corollary.~\ref{homotopy_equiv}, we get from Propositions.~\ref{property_complex} and \ref{homotopy_to_iso} the desired result.
\end{proof}

\begin{proposition}\label{contractible_decomp}
    A bounded contractible complex $C_\bullet$ of projective $R$-modules is isomorphic to a direct sum of finitely many disks of the form $0 \to P \xrightarrow[]{\mathrm{id}} P \to 0$ where $P$ is projective.
\end{proposition}

\begin{proof}
From contractibility, we have the existence of a collection of maps $\{c_n\}$ such that $cd + dc = \mathrm{id}_{C_\bullet}$. At $C_{n}$, we can restrict to $Z_{n}$ so this becomes $d_{n+1} c_{n} \vert_{Z_{n}} = \mathrm{id}_{Z_{n}}$. Thus $B_{n} = \mathrm{im}\, d_{n+1} \supseteq Z_{n}$ and thus $C_\bullet$ is exact. Consider then the short exact sequence:
    \begin{align}
        0 \to Z_n \xrightarrow{\iota} C_n \xrightarrow{d_n} Z_{n-1} \to 0
    \end{align}
    As $d_n c_{n-1} \vert_{Z_{n-1}} = \mathrm{id}_{Z_{n-1}}$, the sequence splits with $C_n \cong Z_n \oplus Z_{n-1}$, and by the splitting lemma, there exists an isomorphism $h_n: C_n \to Z_n \oplus Z_{n-1}$ such that $h_n \iota = \left[\begin{smallmatrix}
        \mathrm{id}_{Z_n} \\ 0
    \end{smallmatrix}\right]$ and $d_n (h_n)^{-1} = \left[\begin{smallmatrix}
        0 & \mathrm{id}_{Z_{n-1}}
    \end{smallmatrix}\right]$. It can be checked that $h_n$ gives a chain map between $C_\bullet$ and the chain complex $D_\bullet$ with $D_n = Z_{n} \oplus Z_{n-1}$ and differential $\left[\begin{smallmatrix}
        0 & \mathrm{id}_{Z_{n-1}} \\ 0 & 0
    \end{smallmatrix}\right]$. As each $Z_n$ is a summand of $C_n$ and thus projective, the claim follows with disks of the form $0 \to Z_n \to Z_n \to 0$ (and as finitely many $Z_n$ are nonzero, there are finitely many non-zero disks).
    
\end{proof}

\begin{proof}[Proof of Lemma.~\ref{tensor_ancilla_symplectic}]
    This follows immediately from Lemma.~\ref{almost_result} and Proposition.~\ref{contractible_decomp}, along with the well-known fact that all finitely-generated projective modules over the Laurent polynomial ring are free \cite{19305ce2-3cfa-35d1-a4a5-425a28d871f5,1976InMat..36..167Q}.
\end{proof}

\clearpage

\noindent\textbf{Proof of Proposition.~\ref{d-dim-complex}} ---
Let $R := \mathbb{F}_2[x_1^{\pm 1}, \dots, x_D^{\pm 1}]$.

\begin{proof}[Proof of Proposition.~\ref{d-dim-complex}]
We first note that the cochain complex (with respect to the $R$-module structure) $\mathrm{Hom}_R(F_\bullet,R)$ is isomorphic to the ZX-dual complex $F^\bullet$ under $f \mapsto \sum_i \overline{f(e_i)} e_i$. We can thus reformulate the conditions of the proposition as $\mathrm{grade}_R(M) \geq D$, where $\mathrm{grade}_R(M)$ measures how long the cochain complex remains exact. By a famous result of Rees'~\cite{Rees_1957}, $\mathrm{grade}_R (R/J) = \mathrm{grade}_R(M)$, where $J = \mathrm{ann}_R (M)$. The former is typically written as $\mathrm{grade}_R (J)$. We have $\mathrm{grade}_R(J) \leq \mathrm{codim}(J)$~\cite{Matsumura_1987} and thus $\mathrm{codim}(J) = D$. It follows that $M$ is a zero-dimensional module.
\end{proof}

\clearpage

\end{document}